\documentclass[journal,twoside,web]{ieeecolor}
\usepackage{lcsys}
\usepackage{cite}
\usepackage{amsmath,amssymb,amsfonts}
\usepackage{algorithmic}
\usepackage{graphicx}
\usepackage{textcomp}

\usepackage{enumerate}
\usepackage{mathrsfs}
\usepackage{float}
\usepackage{cite}
\usepackage{tikz}
\usetikzlibrary{matrix,arrows,decorations.pathmorphing}
\usepackage{verbatim}
\usepackage{mathtools}
\usepackage{caption}
\usepackage{subcaption}
\usepackage{soul}
\setstcolor{magenta}
\usetikzlibrary{arrows}
\usepackage{tabularx}

\newtheorem{cor}{Corollary}

\DeclareMathOperator*{\sgn}{\mathsf{sgn}}
\DeclareMathOperator*{\Dg}{\mathsf{diag}}
\DeclareMathOperator*{\argmax}{argmax}
\graphicspath{{./figs/}}

\newtheorem{remark}{Remark}
\newtheorem{lemma}{Lemma}
\newtheorem{thm}{Theorem}

\def\BibTeX{{\rm B\kern-.05em{\sc i\kern-.025em b}\kern-.08em
    T\kern-.1667em\lower.7ex\hbox{E}\kern-.125emX}}
\begin{document}
\title{On discrete-time polynomial dynamical systems on hypergraphs}
\author{Shaoxuan Cui, Guofeng Zhang \IEEEmembership{Member, IEEE}
, Hildeberto Jardón-Kojakhmetov and Ming Cao \IEEEmembership{Fellow, IEEE}
\thanks{S.Cui was supported by China Scholarship Council. 
G. Zhang was supported by the Guangdong Provincial Quantum Science Strategic Initiative No. (GDZX2200001) and the National Natural Science Foundation of China (No. 62173288). 
M. Cao was supported in part by the Netherlands Organization for Scientific Research (NWO-Vici-19902).
}
\thanks{S. Cui and H. Jard\'on-Kojakhmetov are with the Bernoulli Institute, University of Groningen, Groningen, 9747 AG Netherlands {\tt\small \{s.cui, h.jardon.kojakhmetov\}@rug.nl} }
\thanks{G. Zhang is with the Department of Applied Mathematics, The Hong Kong Polytechnic University, Kowloon 999077, Hong Kong, China 
{\tt\small guofeng.zhang@polyu.edu.hk}}
\thanks{M. Cao is with the ENTEG, University of Groningen, Groningen, 9747 AG Netherlands {\tt\small m.cao@rug.nl}}}

\maketitle
\thispagestyle{empty}
\begin{abstract}
This paper studies the stability of discrete-time polynomial dynamical systems on hypergraphs by utilizing the Perron–Frobenius theorem for nonnegative tensors with respect to the tensors' Z-eigenvalues and Z-eigenvectors. Firstly, for a multilinear polynomial system on a uniform hypergraph, we study the stability of the origin of the corresponding systems. Next, we extend our results to non-homogeneous polynomial systems on non-uniform hypergraphs. We confirm that the local stability of any discrete-time polynomial system is in general dominated by pairwise terms. Assuming that the origin is locally stable, we construct a conservative (but explicit) region of attraction from the system parameters. Finally, we validate our results via some numerical examples.
\end{abstract}

\begin{IEEEkeywords}
Hypergraphs, Higher-order interactions, Polynomial systems, Z-eigenvalues, Perron-Frobenius Theorem, Stability
\end{IEEEkeywords}

\section{Introduction}
\IEEEPARstart{M}{any}  
complex systems such as those originating from epidemics \cite{pare2018analysis,liu2020stability,cui2022discrete}, biology \cite{duarte1998dynamics,slavik2020lotka,goh1976global,goh1979stability}, and engineering \cite{donnell2013local,craciun2019polynomial,angeli2009tutorial,ji2021autonomous} are usually modeled as polynomial systems and studied from a network perspective. However, a conventional network only captures pairwise interaction, 
and may lose some higher-order information \cite{chen2021controllability,chen2020tensor,wolf2016advantages,bick2023higher,wang2024algebraic}. Nowadays, there is abundant evidence that hypergraphs, a generalization of graphs where each edge contains multiple ($\geq 2$) nodes, is a more powerful modeling tool because 
hypergraphs can capture higher-order information \cite{chen2020tensor,wolf2016advantages,bick2023higher,wang2024algebraic}. 
For example, an epidemics model on a hypergraph takes multi-body interactions into account \cite{iacopini2019simplicial,cisneros2021multigroup,cui2023general} and is more suitable to {describe} the process of information diffusion. Similarly, a higher-order Lotka-Volterra model is proposed by \cite{letten2019mechanistic} and studied mathematically by \cite{cui2023species}. This higher-order Lotka-Volterra model takes indirect higher-order effects among species into consideration, which better reflects the reality according to many empirical studies \cite{mayfield2017higher,abrams1983arguments}. All such higher-order systems \cite{cisneros2021multigroup,cui2023general,letten2019mechanistic,cui2023species} are, in fact, non-homogeneous polynomial systems \cite{craciun2019polynomial}.

Although the majority of real systems evolve in continuous-time, a discrete-time system, as an approximation of {its} continuous-time counterpart, still attracts much attention. In control engineering, a controller {may need to be implemented in digital} hardware \cite{aastrom2013computer}. In epidemics, data may be gathered daily and thus it is 
convenient to use a discrete-time model \cite{pare2018analysis,liu2020stability,cui2022discrete}. All these factors motivate us to study a general discrete-time polynomial system. So far, many related works on discrete-time polynomial systems \cite{saat2017analysis,tanaka2008sum} rely on the Sum-Of-Squares (SOS) decomposition
\cite{parrilo2000structured}. The main technique is to find a Lyapunov function in the form of a sum of squares. Generally, the problem of checking the nonnegativity of a function is NP-hard \cite{saat2017analysis}. Some other researches rely on 
Kronecker products and Linear Matrix Inequalities \cite{mtar2009lmi,belhaouane2010improved}. Another approach
uses semi tensor product to study the problem \cite{cheng2007global}. All these mentioned results may be useful for designing a controller to stabilize a polynomial system. However, such results are limited when one wants to know some characteristics of the open-loop system,
such as the domain of attraction of a locally stable equilibrium. Knowing the behavior of an open-loop autonomous system is often of great significance if we deal with real systems like epidemics or species populations, where it is very difficult to implement a controller.

From a modeling perspective, it is well-known that a graph can be captured by its adjacency matrix. For the analysis of polynomial systems on a graph, the Perron–Frobenius theorem \cite{FB-LNS}, which shows that an irreducible nonnegative matrix has a positive eigenvalue associated with a positive eigenvector, is a fundamental tool. 
With the development of tensor algebra, tensor versions of the classical Perron–Frobenius theorem are provided by, e.g., \cite{chang2013survey,chang2008perron,yang2010further,yang2011further} concerning both H-eigenvalues and Z-eigenvalues. So, since a hypergraph can be represented by an adjacency tensor \cite{bick2023higher,gallo1993directed}, it is reasonable to consider a tensor version of the Perron–Frobenius theorem as a potential tool. 
Z-eigenvalues and tensor decomposition have already been used to study the stability of a discrete-time homogeneous polynomial system \cite{chen2021stability}.
For continuous-time polynomial systems on hypergraphs, by using the Perron–Frobenius theorem, the global stability (e.g. of the origin) can be checked by the tensors' H-eigenvalues and H-eigenvectors \cite{cui2024metzler}. It is fair to expect a similar result for a discrete-time system.

\emph{The contributions of the paper are summarized as follows:} first, regarding discrete-time homogeneous polynomial systems on a uniform hypergraph, we achieve a similar result as in \cite{chen2021stability} with a different proof by using a tensor version of the Perron–Frobenius theorem (see Section \ref{sec:uniform}). Next, concerning discrete-time non-homogeneous polynomial systems on non-uniform hypergraphs, under a mild condition, and after knowing an equilibrium is stable, we show that a domain of attraction even though conservative, can be directly calculated from the system's parameters. Lastly, we design some feedback controllers that can reduce or enlarge the aforementioned domain of attraction. To validate our analytical results, we also provide some numerical examples.

\emph{Notation:} $\mathbb{R}$ ($\mathbb{R}_{++}$) denotes the set of  (positive) real numbers. The superscript in e.g. $\mathbb{R}^n$ denotes the dimension of the space. For a matrix $M \in \mathbb{R}^{n \times r}$ and a vector $a \in \mathbb{R}^n$, $M_{ij}$ and $a_{i}$ denote the element in the $i$th row and $j$th column and the $i$th entry, respectively. Given a square matrix $M \in \mathbb{R}^{n \times n}$, $\rho(M)$ denotes the spectral radius of $M$, which is the largest absolute value of the eigenvalues of $M$. The notation $|M|$ denotes the matrix whose entry $|M|_{ij}$ is the absolute value of $M_{ij}$.  
For any two vectors $a, b \in \mathbb{R}^n$, $a \geq (\leq) b$ represents that $a_i \geq(\leq) b_i$, for all $i=1,\ldots,n$.
These component-wise comparisons are also used for matrices or tensors with the same dimensions. The vector $\mathbf{1}$ ($\mathbf{0}$) represents the column vector or matrix of all ones (zeros) with appropriate dimensions. The previous notations have straightforward extensions to tensors. 

\section{Preliminaries on tensors and hypergraphs}
A tensor $T\in\mathbb{R}^{n_1\times n_2 \times \cdots \times n_k}$ is a multidimensional array, where the order is the positive integer $k$ and each dimension $n_i$, $i=1,\cdots,k$ is a mode of the tensor. A tensor is cubical if every mode has the same size, that is $T\in\mathbb{R}^{n\times n \times \cdots \times n}$. We further write a $k$-th order $n$-dimensional cubical tensor as $T\in\mathbb{R}^{n\times n \times \cdots \times n}=\mathbb{R}^{[k,n]}$. A cubical tensor $T$ is called supersymmetric if $T_{j_1 j_2 \ldots j_k}$ is invariant under any permutation of the indices.
For the rest of the paper, a tensor always refers to a cubical tensor.

We then consider the following notation:
for a tensor $A\in\mathbb R^{[k,n]}$, $A x^{k-1}$ is a vector, whose $i$-th component is
$$
\begin{aligned}
\left(A x^{k-1}\right)_i & =\sum_{i_2, \ldots, i_k=1}^n A_{i, i_2 \cdots i_m} x_{i_2} \cdots x_{i_k}.
\end{aligned}
$$
For a tensor $A\in\mathbb R^{[k,n]}$, consider the following equations:
\begin{equation}\label{eq:eigenproblem}
\begin{split}
    A x^{k-1}=\lambda x,\; \qquad
    x^\top x=1.
\end{split}
\end{equation}
If there is a real number $\lambda$ and a nonzero real vector $x$ that satisfy \eqref{eq:eigenproblem}, then $\lambda$ is called a Z-eigenvalue of $A$ and $x$ is the Z-eigenvector of $A$ associated with $\lambda$ \cite{qi2005eigenvalues,chang2013survey}. Throughout this paper, the words eigenvalue and eigenvector as well as Z-eigenvalue and Z-eigenvector are used interchangeably. It is worth mentioning that there are some different kinds of definitions of the eigenvalues and eigenvectors of a tensor, e.g. H-eigenvalues \cite{chang2013survey,qi2005eigenvalues} and U-eigenvalues \cite{zhang2020iterative}. However, in this paper, we will always refer to Z-eigenvalues and Z-eigenvectors as defined above.

The tensor $\mathcal{I}_z=\left(e_{i_1 \ldots i_{k}}\right) \in \mathbb{R}^{[k, n]}$ denotes the Z-identity tensor \cite{mo2021nonnegative} defined as a nonnegative tensor such that
$
\mathcal{I}_z x^{k-1}=x
$
for all $x \in \mathbb{R}^n$ with $x^{\top} x=1$. There is no odd-order Z-identity tensor, i.e. an identity tensor is necessarily of even order. Moreover, Z-identity tensors are not unique\cite{mo2021nonnegative,qi2005eigenvalues}.

A tensor $\mathcal{C}=\left(\mathcal{C}_{{i_1} \ldots {i_{{k}}}}\right) \in \mathbb{R}^{[{k}, n]}$ is called reducible if there {is} a nonempty proper index subset $I \subset\{1, \ldots, n\}$ such that
$$
\mathcal{C}_{i_1 \cdots i_{{k}}}=0 \quad \forall i_1 \in I, \quad \forall i_2, \ldots, i_{{k}} \notin I .
$$
If $\mathcal{C}$ is not reducible, then we call $\mathcal{C}$ irreducible. A tensor with all non-negative entries is called a non-negative tensor.

We now recall the Perron–Frobenius Theorem for irreducible nonnegative tensors with respect to the Z-eigenvalue:

\begin{lemma}[Theorems 4.5 and 4.6 \cite{chang2013survey}]\label{lem:perron}
If $A \in \mathbb{R}_{++}^{[{{k}}, n]}$, then there exists a Z-eigenvalue $\lambda_0 \geq 0$ and a nonnegative Z-eigenvector $x_0 \neq \mathbf{0}$ of $A$ such that $A x_0^{{{k}}-1}=\lambda_0 x_0$. We call $\lambda_0, x_0$ the Perron-Z-eigenvalue and -eigenvector respectively, and refer to $(\lambda_0,x_0)$ as a Perron-Z-eigenpair.
If, in addition, $A \in \mathbb{R}_{++}^{[{{k}}, n]}$ is irreducible, 
then {a} Perron-Z-eigenpair $\left(\lambda_0, x_0\right)$ further satisfies
\begin{itemize}
    \item [1.]  The eigenvalue $\lambda_0$ is positive.
    \item [2.]  The eigenvector $x_0$ is positive{, i.e. $x_0>\mathbf{0}$.}
\end{itemize}
\end{lemma}
It is worthwhile {mentioning} that the Perron-Z-eigenpairs $\left(\lambda_0, x_0\right)$ are generally not unique. Next, we summarize some definitions regarding hypergraphs {as introduced in \cite{gallo1993directed}}.
 
A weighted and directed hypergraph is a triplet $\mathbf{H}=(\mathcal{V},\mathcal{E}, A)$. The set $\mathcal{V}$ denotes a set of vertices and $\mathcal{E}=\{E_1, E_2, \cdots,E_n \}$ is the set of hyperedges. A hyperedge is 
an ordered pair $E=(\mathcal{X},\mathcal{Y})$ of disjoint subsets of vertices; $\mathcal{X}$ is the tail of $E$ and $\mathcal{Y}$ is the head.
As a special case, a weighted and undirected hypergraph is a triplet $\mathbf{H}=(\mathcal{V},\mathcal{E}, A)$, where $\mathcal{E}$ is a finite collection of non-empty subsets of $\mathcal{V}$ \cite{bick2023higher}. If all hyperedges of the hypergraph contain the same number of (tails, heads) nodes, then the hypergraph is uniform. {For more details see \cite{chen2021controllability,xie2016spectral}.}
{From a modeling perspective,} one directed hyperedge usually denotes the joint influence of a group of agents on one agent. Thus, it suffices to deal with hyperedges with one single tail and we assume that each hyperedge has only one tail but possibly multiple ($\geq 1$) heads. This setting is similar to \cite{xie2016spectral} and has the advantage that a directed uniform hypergraph can be represented by an adjacency tensor. Generally, an undirected uniform hypergraph can be represented by a supersymmetric adjacency tensor. 
For a non-uniform hypergraph, we now use the set of tensors $A=\{A_2, A_3,\cdots\}$ to collect the weights of all hyperedges, where $A_2=[A_{ij}]$ denotes the weights of all second-order hyperedges, $A_3=[A_{ijk}]$ denotes the weights of all third-order hyperedges, and so on. For instance, $A_{ijkl}$ denotes the weight of the hyperedge where $i$ is the tail and $j,k,l$ are the heads. For simplicity, in this paper, we also use the weight (for example, $A_{\bullet}$) to denote the corresponding hyperedge. If all hyperedges only have one tail and one head, then the network is a standard directed and weighted graph.
For convenience, we define a multi-index notation $I=i_2,\cdots,i_k$, where $k$ is the order of the associated tensor.

{\section{Discrete-time polynomial dynamical systems on hypergraphs}}

{In this section, we give a brief introduction to the modeling framework of dynamics on a hypergraph. A coupled cell system \cite{stewart2003symmetry} is a network of dynamical systems, or ``cells'', coupled together. Such systems can be represented by a directed network whose nodes correspond to cells and whose edges represent couplings.
Continuous-time coupled cell
systems on a hypergraph are proposed in \cite{bick2023higher} (equation 5.3). If we discretize by using the Euler method, then the corresponding discrete-time counterpart reads as:
\begin{equation}\label{eq:cell}
\begin{split}
    {x}_i^+&=x_i+hF\left(x_i\right)+h\sum_{j\neq i} (A_2)_{ij} G_i\left(x_i, x_j\right)\\
    &+h\sum_{(j, l)\neq (i,i)} (A_3)_{ij l} G_i^{(3)}\left(x_i, x_j, x_l\right)+\cdots,
\end{split}
\end{equation}
\noindent where $x=x(t)\in\mathbb{R}^n$ is the state variable, $h$ is the sampling period, the time instant $t$ is omitted without ambiguity and the shorthand $x^+$ denotes $x(t+1)$; the function $F$ represents the intrinsic coupling of the node $i$, the adjacency matrix $A_2$ together with the coupling functions $G_i$ describe the pairwise network interactions, and the coefficients of $A_s$ (which are adjacency tensors) and coupling function $G_i^{(s)}$ (with $s \geq 3$) are higher-order network interactions. For example, $(A_3)_{ij l}$ and $G_i^{(3)}\left(x_i, x_j, x_l\right)$ describe the joint influence of nodes $l, j$ on node $i$, which can be captured by a directed hyperedge. Next, we further consider that the intrinsic coupling $F$ can be represented by the sum of self-arcs $(A_s)_{i\cdots i}x_i^{s-1}$ and that the coupling functions are of the form $G_{\bullet}^{(s)}(x_i, x_j, x_l,\cdots)=x_i x_j x_l\cdots$. Similar to \cite{chen2021controllability}, all the interactions are characterized by multiplications, which often stand for simultaneity (e.g. the probability of two independent events happening simultaneously). This form of interaction is fairly common, for example, in the SIS epidemic model on a hypergraph \cite{cisneros2021multigroup,cui2023general} and a higher-order Lotka-Volterra model on a hypergraph \cite{cui2023species}, where all interactions are 
multiplicative. In addition, according to the mass action principle, the interactions of a chemical reaction are also multiplicative \cite{craciun2019polynomial}. Moreover, higher-order additive interactions are reduced to combinations of pair-wise interactions, while multiplicative ones are much more general, and cannot be reduced to two-body interactions \cite{cui2023general}.
We see that \eqref{eq:cell} can be written in the tensor form 
\begin{equation}\label{eq:tensor}
    {x}^+=A_{k} x^{k-1}+A_{k-1} x^{k-2}+\cdots +\Tilde{A}_{2} x,
\end{equation}
where $\Tilde{A}_2=A_2+\Dg(x)$, and $h$ is omitted since it can be plugged into the tensors. We emphasize that any polynomial system can be written in this form.}

{We further note that some concrete systems, 
can be rewritten as discrete-time coupled cell systems on a hypergraph. The continuous-time simplicial SIS model on a hypergraph is proposed in \cite{cisneros2021multigroup} (equation 4). 
Similar to the derivation of the discrete-time SIS model on a graph \cite{cui2022discrete}, by applying the Euler method, one gets a discrete-time SIS model on a hypergraph:
\begin{equation}\label{eq:sis}
\begin{split}
        x_i^+&=(1-h\gamma_i) x_i+h\beta_1\left(1-x_i\right) \sum_{j=1}^n a_{i j} x_j\\
&\quad+h\beta_2\left(1-x_i\right) \sum_{j, l=1}^n b_{i j l} x_j x_l.
\end{split}
\end{equation}
where $x_i$ is the infection level of the agent $i$, $\gamma_i$ is the healing rate of the agent $i$; $\beta_1$ is the first order infection rate and $\beta_2$ is the second order infection rate; $a_{i j}$ is the contact rate between $i,j$ and $b_{i j l}$ is the contact rate between $i$ and a group consisting of $j,l$. Both $a_{i j}$ and $b_{i j l}$ correspond to the social contact network, which is a hypergraph. The infection process of the SIS model is briefly presented in Figure \ref{fig:sis}. For more details, see \cite{cisneros2021multigroup,cui2023general}
and observe that  \eqref{eq:sis} can be represented as 
\eqref{eq:cell} and thus also as the tensor form \eqref{eq:tensor}. 
}

\begin{figure}
    \centering
    \includegraphics[height=3cm]{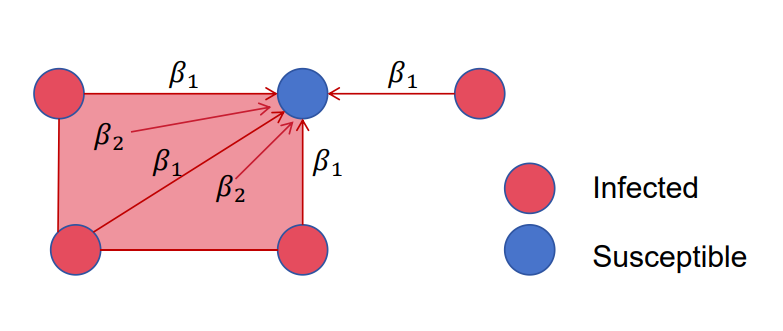}
    \caption{{Illustration of the infection process. The infection rate of $\beta_1$ with normal edges provides classical pairwise interactions, while the infection rate of $\beta_2$ with hyperedges of three elements provides higher-order group-wise interactions. }}
    \label{fig:sis}
\end{figure}

\section{Discrete-time polynomial systems on uniform hypergraphs}\label{sec:uniform}
Here, we consider a discrete-time polynomial system on a uniform hypergraph of $n$ nodes
given by
\begin{equation}\label{eq:sys1}
    x^+=Ax^{k-1},
\end{equation}
where $A\in\mathbb{R}^{[{k}, n]}$. Component-wise, \eqref{eq:sys1} reads as
\begin{equation}
    x^+_i=\sum_{i_2, \ldots, i_k=1}^n A_{i, I} x_{i_2} \cdots x_{i_k}.
\end{equation}

Now, we are ready to discuss the stability of the origin.
\begin{thm}\label{thm:main1}
    The origin is always an equilibrium of \eqref{eq:sys1}. Moreover, if the tensor $A$ is irreducible, the origin is asymptotically stable with a domain of attraction $\max_j \frac{|x_j(0)|}{\delta_j}< (\frac{1}{\lambda})^{\frac{1}{k-2}}$, where 
    {$(\lambda, \delta)$ is a Perron-Z-eigenpair of $|A|$ and {$x(0)=(x_1(0),\ldots,x_n(0))$ is the initial condition}.}
\end{thm}

\begin{proof}
    The first claim is straightforward. {Regarding stability, we}
    define the Lyapunov function $V=\max_j \frac{|x_j|}{\delta_j}$.
    Since the Perron-eigenvector is strictly positive, it holds that $V>0$ for any $x\neq \mathbf{0}$, and $V=0$ if and only if $x=\mathbf{0}$. Furthermore, for any $i$, we have
    \begin{equation}\label{eq:ineq}
        x_i\leq \max_j \left(\frac{x_j}{\delta_j}\right)\delta_i=V\delta_i.
    \end{equation}

We suppose that at time $t$, $q=\argmax_j\left(\frac{|x_j|}{\delta_j}\right)$ and that at time $t+1$, $p=\argmax_j\left(\frac{|x^+_j|}{\delta_j}\right)$. Then, we get 
{\small\begin{equation}
    \begin{split}
        V^+-V &= \frac{1}{\delta_p}(|\sum_{i_2, \ldots, i_k=1}^n A_{p, I} x_{i_2} \cdots x_{i_k}|)-V\\
        &\leq \frac{1}{\delta_p}(\sum_{i_2, \ldots, i_k=1}^n |A_{p, I}| |x_{i_2}| \cdots |x_{i_k}|)-V\\
        &\leq \frac{1}{\delta_p}(\sum_{i_2, \ldots, i_k=1}^n |A_{p, I}| V^{k-1} \delta_{i_2} \cdots \delta_{i_k})-V\\
        &= \frac{1}{\delta_p}(V^{k-1} \lambda \delta_p)-V=V(V^{k-2} \lambda-1),
    \end{split}
\end{equation}}

\noindent where we used \eqref{eq:ineq} to obtain the third line. Once there is a time $t$ such that $\max_j \frac{|x_j(t)|}{\delta_j}< (\frac{1}{\lambda})^{\frac{1}{k-2}}$, then $V^+< V$ and thus $\max_j \frac{|x_j(t+1)|}{\delta_j}\leq V < (\frac{1}{\lambda})^{\frac{1}{k-2}}$. {Analogously}, $V(t+2)\leq V(t+1)$. {By induction, the origin is asymptotically stable with a domain of attraction $\max_j \frac{|x_j(0)|}{\delta_j}< (\frac{1}{\lambda})^{\frac{1}{k-2}}$.}
\end{proof}

\begin{remark}
    The domain of attraction in Theorem \ref{thm:main1} 
    may be conservative. In \cite{chen2021stability}, if $A$ is orthogonally decomposable, a further explicit solution of system \eqref{eq:sys1} is obtained. The result of Theorem \ref{thm:main1} is similar to \cite[Proposition 3]{chen2021stability}. 
    {As described above, discrete-time systems are often obtained by discretizing a continuous-time system. However, even if the latter is homogeneous, its discretization must not be homogeneous, restricting the applicability of Theorem \ref{thm:main1}, and motivating the following section.}
\end{remark}

\section{Discrete-time polynomial systems on non-uniform hypergraphs}
Usually, systems on a non-uniform hypergraph incorporate tensors of different orders. {So, let us consider} 
\begin{equation}\label{eq:sysnon}
    x^+=A_{k-1} x^{k-1}+A_{k-2} x^{k-2}+\cdots +A_{1} x,
\end{equation}
{where $A_{i-1}\in\mathbb R^{[i,n]}$ for $i=k,\ldots,2$.} 
{We note that some of the non-leading tensors $A_i$, $i\neq k-1$, may be zero.}

\begin{thm}\label{thm:nonuni}
    The origin is always an equilibrium of \eqref{eq:sysnon}. Suppose that all the {non-zero} tensors $|A_{k-1}|,\cdots,|A_1|$ are irreducible 
    and have a common Perron-Z-eigenvector $\delta$ and let the corresponding Perron-Z-eigenvalue be $\lambda(|A_{k-1}|),\cdots, \lambda(|A_{1}|)$. If {all the eigenvalues $\Tilde\lambda(A_1)$ of $A_1$ satisfy $|\Tilde{\lambda}(A_1)|<1$}, {where $\Tilde{\lambda}(A_1)$ denote an arbitrary eigenvalue of $|A_1|$, and not necessarily a Perro-Frobenius one}, then the origin is locally asymptotically stable.
   If $\lambda(|A_1|)<1$, then the origin is asymptotically stable with a domain of attraction $\max_j \frac{|x_j(0)|}{\delta_j}< y_+$, where {$y_+\in\mathbb R_{++}$} is the unique positive solution of 
   $\sum_{i=1}^{k-1} \lambda(|A_i|) y^{i-1}=1$.
\end{thm}

\begin{proof}
{Since the Jacobian of \eqref{eq:sysnon} at the origin is given by $A_1$ the local stability of the origin when $|\tilde\lambda(A_1)|<1$ follows from standard theory. Next, define the Lyapunov function $V=\max_j \left(\frac{|x_j|}{\delta_j}\right)$. We suppose that at time $t$, $q=\argmax_j\left(\frac{|x_j|}{\delta_j}\right)$ and that at time $t+1$, $p=\argmax_j\left(\frac{|x^+_j|}{\delta_j}\right)$. Then we get} 
{\small\begin{equation}
    \begin{split}
        &V^+-V = \frac{1}{\delta_p}\Big(\Big|\sum_{i_2, \ldots, i_k=1}^n (A_{k-1})_{p, I} x_{i_2} \cdots x_{i_k}+\\
        &+\sum_{i_3, \ldots, i_k=1}^n (A_{k-2})_{p, I} x_{i_3} \cdots x_{i_k}+\cdots\Big|\Big)-V\\
        &\leq \frac{1}{\delta_p}(\sum_{i_2, \ldots, i_k=1}^n |(A_{k-1})_{p, I}| |x_{i_2}| \cdots |x_{i_k}|+\cdots)-V\\
        &\leq \frac{1}{\delta_p}(V^{k-1} \lambda(|A_{k-1}|) \delta_p+V^{k-2} \lambda(|A_{k-2}|) \delta_p+\cdots)-V\\
        &=V(\sum_{i=1}^{k-1} \lambda(|A_i|) V^{i-1}-1).
    \end{split}
\end{equation}}
{Consider the function $f(y)=\sum_{i=1}^{k-1} \lambda(|A_i|) y^{i-1}-1=0$, we shall show that $f(y)=0$ has a unique positive solution. Note that $f(y)$ is continuous,  $f(0)=\lambda(A_1)-1<0$, and $\lim_{y\rightarrow\infty} f(y)>0$, since all $\lambda(|A_i|)$ are positive.  Then, by the intermediate value theorem, there exists at least one positive solution of $f(y)=0$. Furthermore, we see that $f'(y)= \sum_{i=2}^{k-1} (i-1)\lambda(|A_i|) y^{i-2}>0$ for all $y>0$. This ensures the uniqueness of the solution. Moreover, if there is a unique positive solution $y_+$ of $\sum_{i=1}^{k-1} \lambda(A_i) y^{i-1}=1$, then the solution set of $\sum_{i=1}^{k-1} \lambda(A_i) y^{i-1}-1<0$ is $\{x|x<y_+\}$. By a similar induction argument as in Theorem \ref{thm:main1}, one completes the proof of the domain of attraction.}
\end{proof}

\begin{remark}
    Let $C$ be an irreducible non-negative matrix and $D$ an arbitrary matrix. {From} Wielandt's theorem \cite{gradshteyn2014table}, if  $|D|\leq C$, then any eigenvalue $\Tilde{\lambda}(D)$ satisfies $|\Tilde{\lambda}(D)|\leq \rho(C)$. 
    {So, if  $C=|A_1|$ and $D=A_1$ then $|\Tilde{\lambda}(A_1)|\leq\lambda(|A_1|)<1$, since in the matrix case the Perron-Z-eigenvalue is the spectral radius.}
    Thus, the condition $\Tilde{\lambda}(A_1)\leq\lambda(|A_1|)<1$ tells us that the pairwise interaction must be stable. Theorem \ref{thm:nonuni} further provides a conservative region of attraction. From the definition of irreducibility, the condition that all the tensors $|A_{k-1}|,\cdots,|A_1|$ are irreducible is equivalent to the condition that all the tensors $A_{k-1},\cdots,A_1$ are irreducible. 
\end{remark}

{A clear disadvantage of Theorem \ref{thm:nonuni} is that all non-zero tensors must share a common Perron-Z-eigenvector. We proceed to relax such a restriction.}

\begin{thm}\label{thm:sysnon2}
    Consider system \eqref{eq:sysnon}. Suppose that all the {non-zero} tensors  $A_{k-1},\cdots,A_{2},A_1$ are irreducible. 
    Then, if $\max_p(\sum_l |A_1|_{p,l})<1$, the origin is asymptotically stable with a domain of attraction $\max_j {|x_j(0)|}< \min_p y_{p+}$, where $y_{p+}$ is the unique positive scalar solution of $\sum_{i=1}^{k-1} (\sum_{i_2,\cdots, i_{k}=1}^n |(A_{i})_{p, I}| y^{i-2})=1$. 
\end{thm}

\begin{proof}
    {Since} $|A_{i}|$ is an irreducible nonnegative tensor, then $|A_{i}|\delta^i>\mathbf{0}$ for any
    $\delta>\mathbf{0}$. 
    {Let} $\delta=(\epsilon,\epsilon,\cdots,\epsilon)^\top$ for an arbitrary $\epsilon>0$. 
    Similar to the proof of Theorem \ref{thm:nonuni}, we define $V=\max_j \left(\frac{|x_j|}{\delta_j}\right)= \frac{\max_j|x_j|}{\epsilon}$, and suppose that at time $t$, $q=\argmax_j\left(\frac{|x_j|}{\delta_j}\right)$ and {that} at time $t+1$, $p=\argmax_j\left(\frac{|x^+_j|}{\delta_j}\right)$. Then, we get 
{\small\begin{equation}
    \begin{split}
        V^+-V &= \frac{1}{\epsilon}\Big(\Big|\sum_{i_2,\cdots ,i_k}^{{n}} (A_{k-1})_{p, I} x_{i_2} x_{i_3}\cdots x_{i_k} \\
        &+\sum_{i_3,\cdots ,i_k}^{{n}} (A_{k-2})_{p, I} x_{i_3}\cdots x_{i_k}+\cdots \Big|\Big)-V\\
        &\leq \frac{1}{\epsilon}\left(\sum_{i_2,\cdots ,i_k}^{{n}} |(A_{k-1})_{p, I}| |x_{i_2}|| x_{i_3}|\cdots |x_{i_k}|\right.\\
        &\left.+\sum_{i_3,\cdots ,i_k}^{{n}} |(A_{k-2})_{p, I}|| x_{i_3}|\cdots |x_{i_k}|+
        \cdots\right)-V\\
        &\leq \frac{1}{\epsilon}\left(\sum_{i_2,\cdots ,i_k}^{{n}} |(A_{k-1})_{p, I}|V^{k-1} \epsilon^{k-1}+\cdots\right.\\
        &\left.+\sum_{i_2}^{{n}} |(A_{2})_{p, I}|V \epsilon\right)-V\\
        &= V\left(\sum_{i=1}^{k-1} \left(\sum_{I} |(A_{i})_{p, I}| (V\epsilon)^{i-2}\right)-1\right).
    \end{split}
\end{equation}}
Using similar arguments as in the proof of Theorem \ref{thm:nonuni}, we know that there is a unique positive solution $y_{p+}$ for the equation $\sum_{i=1}^{k-1} (\sum_{i_2, \cdots,i_{k}=1}^n |(A_{i})_{p, I}| y^{i-2})=1$ if $\max_p(\sum_l |A_1|_{p,l})<1$.
Making $V^+-V <0$ requires (i): $V \epsilon=\max_j|x_j|< y_{p+}$ and (ii): $\sum_{i2=1}^n |(A_{k-2})_{p, I}| -1<0$. Notice that the index $p$ may change from time to time. However, once $\max_j {|x_j(0)|}<y_{p+},\forall p=1,\cdots,n$ and $\sum_l |A_1|_{p,l}<1, \forall p=1,\cdots,n$, then the conditions (i) and (ii) are satisfied. This completes the proof.
\end{proof}

\begin{remark}
    Theorem \ref{thm:sysnon2} is very useful for potential {applications since it} just requires that each tensor $A_i$ is irreducible. Moreover, $y_{p+}$ is only related to the absolute value of the entries of each tensor $A_i$. If we know all {the} tensors, it is easy to compute the domain of attraction. Furthermore, from the Gershgorin circle theorem, the condition $\max_p(\sum_l |A_1|_{p,l})<1$ guarantees that $|\Tilde{\lambda}(A_1)|\leq\lambda(|A_1|)<1$ and thus the origin must be locally stable.
\end{remark}

Next, let us consider the following quadratic system:
\begin{equation}\label{eq:sysqua}
    x^+=A_{2} x^{2}+A_{1} x.
\end{equation}
Using Theorem \ref{thm:sysnon2} we have a more concrete result:

\begin{cor}\label{cor:qua}
    The origin is always an equilibrium of \eqref{eq:sysqua}. Suppose that both tensors $A_{2},A_1$ are irreducible. Then, if $\max_p(\sum_l |A_1|_{p,l})<1$, the origin is asymptotically stable with a domain of attraction $\max_j {|x_j(0)|}< \min_p \frac{1-\sum_l|A_1|_{p,l}}{\sum_{l,m}|A_2|_{p,lm}}$.
\end{cor}

Similarly, we can look at the cubic system:
\begin{equation}\label{eq:syscub}
    x^+=A_{3} x^{3}+A_{2} x^{2}+A_{1} x.
\end{equation}

\begin{cor}\label{cor:cub}
    The origin is always an equilibrium of \eqref{eq:syscub}. Suppose that all the tensors $A_3,A_{2},A_1$ are irreducible. Then, if $\max_p(\sum_l |A_1|_{p,l})<1$, the origin is asymptotically stable within a domain of attraction $\max_j {|x_j(0)|}< \min_p \frac{-C_3+\sqrt{C_2^2-4C_3(C_1-1)}}{2C_3}$, where $C_3=\sum_{l,m,q}|A_3|_{p,lmq},\,C_2=\sum_{l,m}|A_2|_{p,lm},\,C_1=\sum_l|A_1|_{p,l}$.
\end{cor}

{In addition, consider polynomial systems on non-uniform hypergraphs with constant terms:}
\begin{equation}\label{eq:sysnon2}
    x^+=A_{k-1} x^{k-1}+A_{k-2} x^{k-2}+\cdots +A_{1} x+b,
\end{equation}
where $b$ is a constant vector. {With the coordinate change $y=x-a$ one maps \eqref{eq:sysnon2} to the form of \eqref{eq:sysnon}. Thus, all the aforementioned results in this section apply to \eqref{eq:sysnon2} as well.}

\section{Feedback control strategies}
In this section, we propose some feedback control strategies for system \eqref{eq:sysnon} to manipulate the conservative domain of attraction in Theorem \ref{thm:nonuni}.

We consider the closed loop system
\begin{equation}\label{eq:sysclsnon}
    x^+=A_{k-1} x^{k-1}+A_{k-2} x^{k-2}+\cdots +A_{1} x+g(u),
\end{equation}
with a feedback controller $g(u)=s\Tilde{\mathcal{I}}x^{l-1}$ where $l>2$ is an even number, 
$\Tilde{\mathcal{I}}\in\mathbb R^{[l,n]}$ that we design, and $s\in\mathbb R$.

We {let} $\Tilde{\mathcal{I}}_{i_1,i_2,\cdots,i_n}=({\mathcal{I}}_z)_{i_1,i_2,\cdots,i_n}$ if $({\mathcal{I}}_z)_{i_1,i_2,\cdots,i_n}=0$; {and} otherwise, $\Tilde{\mathcal{I}}_{i_1,i_2,\cdots,i_n}=\sgn((A_l)_{i_1,i_2,\cdots,i_n})|({\mathcal{I}}_z)_{i_1,i_2,\cdots,i_n}|$ 
We assume that $\max_{i_1,i_2,\cdots,i_n} |A_l|_{i_1,i_2,\cdots,i_n}>|s|$ if $s<0$. By this construction, we have that $(|A_l+s\Tilde{\mathcal{I}}|)_{i_1,i_2,\cdots,i_n}=(|A_l|)_{i_1,i_2,\cdots,i_n}+s({\mathcal{I}}_z)_{i_1,i_2,\cdots,i_n}$ such that $|A_l+s\Tilde{\mathcal{I}}|=|A|+s{\mathcal{I}}_z$. Then, $(|A_l|+s{\mathcal{I}}_z) \delta^{k-1}=(\lambda(|A_l|)+ s)\delta$, where $(\lambda,\delta)$ is a Perron-Z-eigenpair of $|A_l|$. This means {that} $\lambda(|A_l|)$ in Theorem \ref{thm:nonuni} is substituted {by} $\lambda(|A_l|)+ s$. Consider the equation $f(y)=\sum_{i=1}^{k-1} \lambda(|A_i|) y^{i-1}-1=0$. If we {increase} (decrease) $\lambda(|A_l|)$ with $s>0$ ($s<0$), {and} since $f(0)=\lambda(|A_1|)-1$ is unchanged but $f^{'}(y)$ increases (decreases) as $\lambda(|A_l|)$ increases (decreases), {then} $y^+$ decreases (increases). In this way, we can manipulate the conservative domain of attraction. 
{Since \eqref{eq:sys1} with an even $k$ is a special case of \eqref{eq:sysnon}, the control strategy just described is also applicable to \eqref{eq:sys1} by choosing $g(u)=s\Tilde{\mathcal{I}}x^{k-1}$.}

\section{Numerical examples and further discussions}
Consider \eqref{eq:sysqua} with $A_1=\left[\begin{array}{cc}
    0.1 & 0.1 \\
    0.1 & 0.1
\end{array}\right]$ and {$A_2\in\mathbb R^{[3,2]}$}, whose entries are all one except $({A}_2)_{112}=({A}_2)_{212}=({A}_2)_{121}=({A}_2)_{221}=0.5$. Component-wise {we have:}
\begin{equation}\label{eq:exa1}
    x_i^+= 0.1x_1+0.1x_2+x_1^2+x_1x_2+x_2^2, \quad i=1,2.
\end{equation}
{
From Corollary \ref{cor:qua}, the domain of attraction is $\max_j {|x_j(0)|}< \min_p \frac{1-\sum_l|A_1|_{p,l}}{\sum_{l,m}|A_2|_{p,lm}}=\frac{4}{15}$. Then, we use Matlab to scatter the domain of attraction for the origin.
}The result is shown in figure \ref{fig:fig1}. We can see that the conservative region of attraction according to the Theorem \ref{thm:sysnon2} is indeed a region of attraction of the origin.
{Next}, let {$\Tilde{A}_2\in\mathbb R^{[3,2]}$ with} entries all one except $(\Tilde{A}_2)_{112}=(\Tilde{A}_2)_{212}=1.5$ and $(\Tilde{A}_2)_{121}=(\Tilde{A}_2)_{221}=-0.5$.
{We see that \eqref{eq:sysqua} with $(A_1,A_2)$ and with $(A_1,\Tilde A_2)$ yield the same component-wise representation \eqref{eq:exa1}.}
However, from Corollary \ref{cor:qua}, for the system with $\Tilde{A}_2$, the conservative region of attraction is $\max_j {|x_j(0)|}< \min_p \frac{1-\sum_l|A_1|_{p,l}}{\sum_{l,m}|A_2|_{p,lm}}=\frac{0.8}{4}=0.2<\frac{4}{15}$ and is smaller than the conservative region of attraction calculated with $A_2$. We see that although one deals with the same system, the choice of tensors  influence{s} the {computation of the} conservative region of attraction. To make the region as large as possible, one needs to choose appropriate higher order ($\geq3$) tensors such that $\sum_{i_2,\cdots, i_{k}=1}^n |(A_{i})_{p, I}|, i\geq 3$ is as small as possible.

\begin{figure}
    \centering
    \includegraphics[height=6cm]{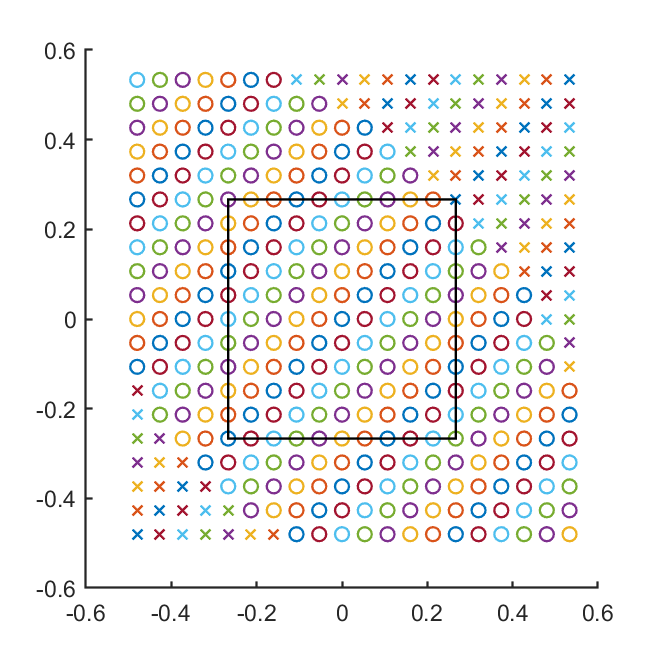}
    \caption{Points with 'o' are within the domain of attraction {of} the origin {while} points with 'x' are out of the domain. 
     The square area (without the boundary) denotes the conservative region of attraction {from} Theorem \ref{thm:sysnon2} and Corollary \ref{cor:qua}.}
    \label{fig:fig1}
\end{figure}

\section{Conclusions}
In this paper, we investigate general discrete-time non-homogeneous polynomial dynamical systems on non-uniform hypergraphs. In particular, we give a simple way to calculate the conservative region of attraction of the origin directly {from the} system{'s} parameters.  Furthermore, we develop a feedback control strategy that can manipulate the conservative region of attraction. Finally, the main results are illustrated via numerical examples.

\bibliographystyle{IEEEtran}
\bibliography{bib}

\end{document}